\documentclass{article}
\usepackage[margin=1in]{geometry}
\usepackage[utf8]{inputenc}

\usepackage{balance}
\usepackage{graphicx}
\usepackage{tikz}
\usepackage{subcaption}
\usepackage{caption}
\usepackage[dvipsnames]{xcolor}
\usetikzlibrary{positioning, shapes.geometric}

\usepackage{amssymb,amsfonts,mathtools, xcolor, comment, amsthm,url}
\usepackage{cleveref}
\usepackage{booktabs}

\newtheorem{lemma}{Lemma}
\newtheorem{assumption}{Assumption}
\newtheorem{definition}{Definition}
\newtheorem{remark}{Remark}
\newtheorem{example}{Example}
\newtheorem{proposition}{Proposition}
\newtheorem{theorem}{Theorem}
\newtheorem{corollary}{Corollary}

\crefname{figure}{Fig.}{Fig.}
\crefname{example}{Example}{Example}
\crefname{proposition}{Proposition}{Proposition}
\crefname{corollary}{Corollary}{Corollary}
\crefname{theorem}{Theorem}{Theorem}
\crefname{lemma}{Lemma}{Lemma}
\crefname{assumption}{Assumption}{Assumption}
\crefname{appendix}{Appendix}{Appendix}
\crefname{definition}{Definition}{Definition}
\crefname{remark}{Remark}{Remark}
\crefname{table}{Table}{Table}
\crefname{example}{Example}{Example}

\newcommand{\tol}{\beta}

\newcommand{\paths}{\mathcal{P}} % set of paths
\title{Average Unfairness in Routing Games}
\author{Pan-Yang Su \thanks{Equal contribution} \thanks{University of California, Berkeley}
\and Arwa Alanqary \footnotemark[1] \footnotemark[2]
\and Bryce L. Ferguson \thanks{Thayer School of Engineering at Dartmouth College}
\and Manxi Wu \footnotemark[2]
\and Alexandre M. Bayen \footnotemark[2]
\and Shankar Sastry \footnotemark[2]
}
\date{}

\begin{document}

\maketitle

%%% Use this environment to specify a short abstract for your paper.
\begin{abstract}
We propose average unfairness as a new measure of fairness in routing games, defined as the ratio between the average latency and the minimum latency experienced by users. This measure is a natural complement to two existing unfairness notions: loaded unfairness, which compares maximum and minimum latencies of routes with positive flow, and user equilibrium (UE) unfairness, which compares maximum latency with the latency of a Nash equilibrium. We show that the worst-case values of all three unfairness measures coincide and are characterized by a steepness parameter intrinsic to the latency function class. We show that average unfairness is always no greater than loaded unfairness, and the two measures are equal only when the flow is fully fair. {Besides that, we offer a complete comparison of the three unfairness measures, which, to the best of our knowledge, is the first theoretical analysis in this direction.} Finally, we study the constrained system optimum (CSO) problem, where one seeks to minimize total latency subject to an upper bound on unfairness. We prove that, for the same tolerance level, the optimal flow under an average unfairness constraint achieves lower total latency than any flow satisfying a loaded unfairness constraint. {We show that such improvement is always strict in parallel-link networks and establish sufficient conditions for general networks. We further illustrate the latter with numerical examples.} Our results provide theoretical guarantees and valuable insights for evaluating fairness-efficiency tradeoffs in network routing. 
\end{abstract}


%%%%%%%%%%%%%%%%%%%%%%%%%%%%%%%%%%%%%%%%%%%%%%%%%%%%%%%%%%%%%%%%%%%%%%%%

\section{Introduction}
In the classical routing games model, user behavior naturally leads to equilibria that can be significantly inefficient compared to centrally optimized solutions. 
While these optimal flows achieve the lowest total latency among the population, such centralized routing can often compromise \textit{fairness}, potentially discriminating against certain users by assigning them disproportionately high latency routes. 
This fundamental tension between efficiency and fairness has motivated a growing body of work aimed at formalizing measures of unfairness, understanding its trade-offs with efficiency, and developing algorithms for computing flows that minimize the total latency under unfairness constraints.

A key concern in this line of research is defining what it means for a flow to be fair. 
Much of the literature adopts worst-case fairness measures, which quantify fairness by the experience of the most disadvantaged user. 
While this perspective provides strong guarantees, it is brittle: a single outlier can dominate the fairness measure, regardless of how many users actually suffer such extreme latency. 
As a result, worst-case fairness measures may be overly pessimistic and fail to reflect the average user experience.

In this work, we introduce a new \textit{average} measure of unfairness that captures the average envy experienced by users in a given flow. 
By shifting the focus from the worst-off user to the expected delay in the network, our approach provides a more stable way to reason about fairness in routing games. 
In light of this newly introduced measure, we revisit the question: \textit{How unfair is optimal routing?} To that end, we prove a tight bound on the average unfairness of an optimal flow for a class of latency functions. 
We also revisit two canonical notions of worst-case unfairness: the \textit{loaded unfairness}, which is the maximum ratio between two users' experienced latency, and the \textit{user equilibrium unfairness}, which is the maximum ratio between a user's latency at the given flow and the latency they experience at equilibrium. 
Interestingly, the bound that we derive for the average unfairness matches the bounds for both of these measures, but the lower bound is attained at different instances. 
Further, for any given instance and a feasible flow exhibiting some level of unfairness, we show that the average unfairness is always strictly less than the corresponding loaded unfairness.  
This highlights that worst-case measures can overstate the level of unfairness and it motivates the use of average-case measures as user constraints.
{Besides that, we offer a complete comparison of all the three unfairness measures. To the best of our knowledge, since their introduction by \cite{angelelli2021system}, these relationships have not been formally studied. Our analysis complements prior work \cite{angelelli2021system}, which investigated these connections only through numerical experiments, by providing the first rigorous theoretical treatment.}
{
Finally, we study the constrained system optimum problem (CSO) under the alternative unfairness measures. For a fixed tolerance, the CSO solution with the average-unfairness constraint can attain a strictly lower cost than the loaded-unfairness solution.
We derive verifiable conditions under which this strict improvement is guaranteed. These conditions are expressed in terms of the path support of the loaded-unfairness CSO solution and always hold in parallel-link networks, ensuring strict improvement in that class. On benchmark networks with general topology, our numerical experiments frequently exhibit the same effect.
}

The main contributions of this paper are summarized as follows.
\begin{enumerate}
    \item We introduce average unfairness as an intuitive and interpretable measure of unfairness and derive its exact bound under general network topologies and latency function classes.
    \item We provide a comprehensive theoretical comparison between average unfairness and two commonly used measures in the literature (loaded and user equilibrium (UE) unfairness). To the best of our knowledge, this is the first theoretical comparison since the numerical study in \cite{angelelli2021system}.
    \item We analyze the constrained system optimum (CSO) problem under all three unfairness measures. We establish the structure of its solutions in parallel-link networks, extending the result in \cite{10.1145/380752.380783}. Moreover, we show that the CSO solution under average unfairness can strictly outperforms that under loaded unfairness, and establish conditions under which this improvement is guaranteed. We further validate these results through simulations.
\end{enumerate}

\subsection{Related Work} 
\textbf{Unfairness measures and constrained system optimum problem.}
The notion of unfairness in routing games and its first formalization were introduced by Jahn et al. \cite{doi:10.1287/opre.1040.0197}. 
This work proposed a menu of unfairness measures each capturing different aspects of unfairness.
One of the most studied among these is the \textit{loaded unfairness}, which measures the maximum ratio between the latencies of used paths. It captures the worst-case envy among users by quantifying how much more latency one user experiences compared to another. This measure has been further analyzed in \cite{doi:10.1287/opre.1070.0383,basu2017reconciling}, and a closely related definition was introduced and analyzed in \cite{jalota2023balancing}. 
An analogous measure, the minmax fairness or envey-ratio, was proposed and analyzed for the atomic congestion games model in \cite{chakrabarty2005fairness,gollapudi2025fairness}.

Another measure, the \textit{user equilibrium unfairness}, captures the deviation from the Nash equilibrium flow. It is defined as the maximum ratio between the latency of a path used in the current flow and that of a path used at equilibrium. This measure was studied in \cite{10.1145/506147.506153} to bound the unfairness of an optimal flow. 
Jahn et al. \cite{doi:10.1287/opre.1040.0197} also introduced the \textit{normal unfairness}, which evaluates fairness based on the ratio of the normal length of a given path to that of the shortest path. The normal length refers to any property of the path that is fixed a priori and is independent
of the flow. While attractive for its simplicity \cite{angelelli2021system}, this measure fails to reflect the actual level of experienced unfairness under realistic routing conditions \cite{angelelli2020minimizing}.
An improved notion of normal unfairness that uses the path's latency at the Nash equilibrium flow as the normal length was proposed and analyzed in \cite{schulz2006efficiency}.
Finally, the \textit{free-flow unfairness}—the least studied of the four—measures the maximum ratio between the latency of a used path and its latency under no congestion. It can be viewed as an intermediary between the normal and loaded unfairness measures.

We note that all these unfairness measures that are analyzed in the context of routing games are worst-case measures, focusing on the extremes of the population. 
In contrast, the average unfairness measure we propose shifts the focus to quantifying the average envy each user feels toward the most-advantaged user in the population. 
This approach aligns with some fairness notions studied in the broader resource allocation literature, where population-level measures are common. 
A notable example is the Gini index \cite{farris2010gini}, which captures average pairwise differences in users’ costs normalized by the average cost. 
A key technical distinction is that the Gini index is based on cost differences, while our measure uses cost ratios. 
Conceptually, while the Gini index reflects the average envy between all pairs of users, our measure captures the average envy of each user relative to the best-off user, providing a perspective more aligned with the previously studied measures in this literature. 

Previous numerical comparisons between these measures \cite{doi:10.1287/opre.1040.0197,jalota2023balancing} highlight the fact that they are not equivalent. 
In this work, we strengthen the understanding of these relationships by establishing theoretical comparisons between these measures, including the newly introduced average unfairness, whenever they are comparable, and illustrate with examples where such relationships do not hold.

The study of unfairness is motivated by the \textit{constrained system optimum (CSO)} problem which aims at finding flows that minimize the total latency among those flows with bounded unfairness.
This problem is of great practical interest as it guides the design of routing algorithms and recommendations. 
Our average-case unfairness measure gives rise a new formulation of the CSO model that is more realistic and potentially with better (more efficient) solutions. 

\textbf{Approximate equilibrium and price of anarchy and stability.}
Unfairness measures are closely related to the concept of approximate equilibria in routing games, though they are motivated by distinct concerns: fairness versus bounded rationality. 
Despite their different motivations, the definitions of measures and solution sets, as well as the analytical tools used in both lines of research are similar.
A key distinction, however, lies in the fact that approximation measures fully characterize the set of exact Nash equilibria as the approximation tolerance goes to zero, a property not shared by any of the unfairness measures commonly studied. 
This creates subtle differences in the way these measures are defined.
For instance, the multiplicative definition of approximate equilibrium \cite{10.1145/506147.506153} compares the cost of a used path to that of any alternative, whereas the loaded unfairness introduced above only compares used paths.
Moreover, performance measures developed for approximate equilibria, such as the price of stability (POS) \cite{christodoulou2011performance}, 
or for tolled equilibria, such as the tolled price of anarchy (POA) \cite{ferguson2021impact},
offer a complementary perspectives useful for the study of fairness-related problems, such as the CSO problem. 
In particular, the POS, with respect to a suitable unfairness/approximation measure, offer a meaningful upper bound on the values of the corresponding CSO problem for a class of routing games. 

\section{Preliminaries}
We provide the preliminaries of a routing game, following the model in \cite{10.1145/506147.506153,10.5555/545381.545406}. 

\subsection{The Model}
We consider a directed network $G = (V, E)$ with vertex set $V$, edge set $E$, and $k$ source-destination vertex pairs $\{s_1, t_1\}, ..., \{s_k, t_k\}$, where each pair is referred to as a \emph{commodity}. We denote the set of (simple) $s_i$-$t_i$ paths by $\mathcal{P}_i$, which is assumed to be non-empty, and define $\mathcal{P} = \cup_i \mathcal{P}_i$. A \emph{flow} is a function $f: \mathcal{P} \rightarrow \mathbb{R}^+$; for a fixed flow $f$ we define $f_e = \sum_{P \in \mathcal{P}: e \in P} f_P$. We associate a finite and positive \emph{rate} $r_i$ with each pair $\{s_i, t_i\}$, the amount of flow with source $s_i$ and destination $t_i$; a flow $f$ is said to be \emph{feasible} if for all $i$, $\sum_{P \in \mathcal{P}_i} f_p = r_i$. Finally, each edge $e\in E$ is given a load-dependent \emph{latency function} that we denote by $\ell_e$. For each $e\in E$, we assume that the latency function $\ell_e$ is nonnegative, continuous, and nondecreasing. Following \cite[Definition 2.5]{10.1145/509907.509971}, the following definition will be useful later.
\begin{definition}
A latency function $\ell$ is standard if $x\cdot \ell(x)$ is convex on $[0,\infty )$.
\end{definition}

We will call the triple $(G, r, \ell)$ an \emph{instance}. The latency of a path $P$ with respect to a flow $f$ is defined as the sum of the latencies of the edges in the path, denoted by $\ell_P(f) = \sum_{e \in P} \ell_e(f_e)$\footnote{Note that $\ell_e(f_e)$ only depends on $f_e$, the amount of traffic on edge $e$, but $\ell_P(f)$ depends not only on $f_P$, the amount of traffic on path $P$, but also on the amount of traffic on each edge $e \in P$. Specifically, $\ell_P(f)$ depends on $\{f_e | e \in P\}$.}. We define the \emph{cost} $C(f)$ of a flow $f$
in $(G, r, \ell)$ as the total latency incurred by $f$; that is,
\begin{equation}
C(f) = \sum_{P \in \mathcal{P}} \ell_P(f) f_P.
\end{equation}
By summing over the edges in a path $P$ and reversing the order of summation, we may also write $C(f) = \sum_{e\in E} \ell_e(f_e)f_e$. Also, the total latency of commodity $i$ is $C_i(f) = \sum_{P \in \mathcal{P}_i} \ell_P(f)f_P$. 

Finally, given a class $\mathcal{L}$ of latency functions, we define a class of instances $\mathcal{G}(\mathcal{L}) = \{(G, r, \ell)| G \in \mathcal{G}^{dir}_{fin}, r \succ 0, \ell_e \in \mathcal{L}\}$, where $\mathcal{G}^{dir}_{fin}$ is the set of finite directed graphs. A class $\mathcal{G}(\mathcal{L})$ represents all routing games on finite directed graphs with latency functions from the class $\mathcal{L}$. 

\subsection{Nash and Optimal Flows}
\label{sec: Nash and Optimal Flows}
Fix a flow $f$ feasible for an instance $(G, r, \ell)$. We say that $f$ is an \emph{optimal flow} if it minimizes $C(f)$. We denote the set of optimal flows by $OPT(G, r, \ell)$. We say that it is a \emph{Nash flow} if for every $i$ and every two $s_i$-$t_i$ paths $P_1, P_2 \in \mathcal{P}_i$ with $f_{P_1} > 0$, $\ell_{P_1}(f) \leq \ell_{P_2}(f)$. With nondecreasing and continuous latency functions, both optimal flows and Nash flows exist\footnote{Continuity alone suffices for the existence of an optimal flow due to the extreme value theorem, but the existence of a Nash flow needs both assumptions; see \cite[Lemma 2.6]{10.1145/506147.506153}.}.

Given an instance $(G, r, \ell)$ with standard latency functions, define the \emph{marginal cost function} $\hat{\ell}_e = \frac{d}{dx}(x\cdot \ell_e(x))$. Then, a flow $\hat{f}$ feasible for $(G, r, \ell)$ is optimal if and only if it is a Nash flow for $(G, r, \hat{\ell})$; that is, for every $i$ and every two $s_i$-$t_i$ paths $P_1, P_2 \in \mathcal{P}_i$ with $f_{P_1} > 0$, $\hat{\ell}_{P_1}(f) \leq \hat{\ell}_{P_2}(f)$.

\subsection{Unfairness}
We revisit two load-dependent unfairness measures introduced in  \cite{doi:10.1287/opre.1040.0197} and define the average unfairness measure. Note that we define $0/0 := 1$, so a flow where all traffic incurs zero latency has an unfairness of 1 (totally fair).
%%%%%%%%%%%%%% Loaded unfairness %%%%%%%%%%%%%%%%%%%
\begin{definition}[Loaded unfairness]
Given an instance $(G, r, \ell)$ and a feasible flow $f$, the \emph{loaded unfairness} $U^{L}(G, r, \ell, f)$ is defined as
\begin{equation}
\label{eq: L def}
U^{L}(G, r, \ell, f) = \max\left\{\frac{\ell_P(f)}{\ell_Q(f)} | P, Q \in \mathcal{P}_i, f_P, f_Q > 0, i \in \{1, 2, ..., k\}\right\}.
\end{equation}
\end{definition}
We denote $U^{L}_i(G, r, \ell, f) = \max\{\frac{\ell_P(f)}{\ell_Q(f)} | P, Q \in \mathcal{P}_i, f_P, f_Q > 0\}$ for commodity $i$'s loaded unfairness. 
For a given instance, we denote its loaded unfairness $U^L(G, r, \ell) := \sup_{f^{\star} \in OPT(G, r, \ell)} U^{L}(G, r, \ell, f^{\star})$.
Given a class of latency functions $\mathcal{L}$, we denote its loaded fairness $U^{L}(\mathcal{L}) = \sup_{(G, r, \ell) \in \mathcal{G}(\mathcal{L})}U^{L}(G, r, \ell)$.

\begin{remark}
In defining $U^L(G, r, \ell)$, we take the supremum rather than the infimum to capture the worst-case scenario, although, in practice, one might seek an optimal flow that minimizes unfairness.
\end{remark}

\begin{remark}
\label{remark:unfairness_unique}
Different flows, particularly optimal flows $f^{\star} \neq \tilde{f}^{\star}$, that induce the same edge flows $f^{\star}_e = \tilde{f}^{\star}_e, \forall e \in E$, can have different values of loaded unfairness $U^L(f^{\star}) \neq U^L(\tilde{f}^{\star})$. 
To illustrate this, consider the network in \cref{fig:pigou-series} which consists of two Pigou networks connected in series. 
{When routing one unit of traffic, any flow that induces traffic of $1/2$ to each edge is optimal, but splitting the unit demand equally between paths 1 and 2 yields
unfairness 2 (worst case unfairness for instances with linear latency functions \cite{doi:10.1287/opre.1070.0383}), while splitting between paths 3 and 4 yields unfairness 1 (totally fair).}
This example shows that even though these flows have the same cost, they induce different levels of unfairness. 
\end{remark}

\begin{figure}[h]
    \centering
    \begin{tikzpicture}[ ->, thick, node distance=2.5cm]
        %==============================
        % Left: Network
        %==============================
        % Nodes for first Pigou network
        \node[circle, draw, minimum size=0.5cm, inner sep=1pt] (A1) {$s$};
        \node[circle, draw, minimum size=0.5cm, inner sep=1pt, right=of A1] (M) {};
        \node[circle, draw, minimum size=0.5cm, inner sep=1pt, right=of M] (B2) {$t$};
        
        % First Pigou network
        \draw (A1) to[bend left=30] node[midway, above] {$x$} (M);
        \draw (A1) to[bend right=30] node[midway, below] {$1$} (M);
        
        % Second Pigou network
        \draw (M) to[bend left=30] node[midway, above] {$x$} (B2);
        \draw (M) to[bend right=30] node[midway, below] {$1$} (B2);
        
        %==============================
        % Right: Paths
        %==============================
        \begin{scope}[yshift=-0.2cm, xshift=8cm]
        % Path 1: x -> x (up-up)
        \node[black] at (0,1) {\small path 1};
        \coordinate (start1) at (1,1);
        \coordinate (mid1) at (2,1);
        \coordinate (end1) at (3,1);
        \draw[->, black] 
        (start1) to[bend left=30] (mid1) 
                 to[bend left=30] (end1);
        % \node[blue, above=2pt] at (1.5,6.2) {\small path 1};
        % Path 2: 1 -> 1 (down-down)
        \node[black] at (0,0.5) {\small path 2};
        \coordinate (start2) at (1, 0.5);
        \coordinate (mid2) at (2, 0.5);
        \coordinate (end2) at (3, 0.5);
        \draw[->, black] 
        (start2) to[bend right=30] (mid2) 
                 to[bend right=30] (end2);
        % \node[blue, above=2pt] at (1.5, 4.8) {\small path 2};
        % Path 1: x -> 1 (up-down)
        \node[black] at (0,0) {\small path 3};
        \coordinate (start3) at (1,0);
        \coordinate (mid3) at (2,0);
        \coordinate (end3) at (3,0);
        \draw[->, black] 
        (start3) to[bend left=30] (mid3) 
                 to[bend right=30] (end3);
        % \node[red, above=2pt] at (2,4.7) {\small $x \rightarrow 1$};
        % Path 3: 1 -> x (down-up)
        \node[black] at (0,-0.5) {\small path 4};
        \coordinate (start4) at (1,-0.5);
        \coordinate (mid4) at (2,-0.5);
        \coordinate (end4) at (3,-0.5);
        \draw[->, black] 
        (start4) to[bend right=30] (mid4) 
                 to[bend left=30] (end4);
        % \node[teal, above=2pt] at (2,2.7) {\small $1 \rightarrow x$};
        \end{scope}
    \end{tikzpicture}
    \caption{Two Pigou networks connected in series.}
    \label{fig:pigou-series}
\end{figure}
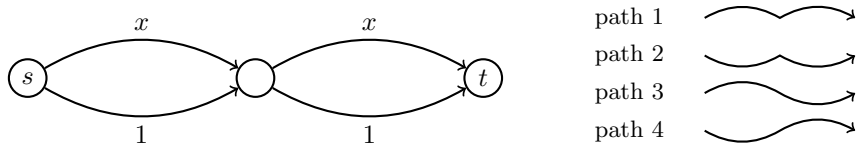

%%%%%%%%%%%%%% UE unfairness %%%%%%%%%%%%%%%%%%%
\begin{definition}[UE unfairness]
Given an instance $(G, r, \ell)$ and a feasible flow $f$, the \emph{user equilibrium (UE) unfairness} $U^{UE}(G, r, \ell, f)$ is defined as
\begin{equation}
\label{eq: UE def}
\begin{aligned}
U^{UE}(G, r, \ell, f) = \max\biggl\{&\frac{\ell_P(f)}{\ell_Q(f^{NE})} | P, Q \in \mathcal{P}_i, f_P, f^{NE}_Q > 0, \ i \in \{1, 2, ..., k\} \biggr\},
\end{aligned}
\end{equation}
where $f^{NE}$ is a Nash flow. 
\end{definition}
We denote $U^{UE}_i(G, r, \ell, f) = \max\{\frac{\ell_P(f)}{\ell_Q(f^{NE})} | P, Q \in \mathcal{P}_i, f_P, f^{NE}_Q > 0\}$ for commodity $i$'s UE unfairness.
For a given instance, we denote its UE unfairness $U^{UE}(G, r, \ell) := \sup_{f^{\star} \in OPT(G, r, \ell)} U^{UE}(G, r, \ell, f^{\star}) $.
Given a class of latency functions $\mathcal{L}$, we denote its UE fairness $U^{UE}(\mathcal{L}) = \sup_{(G, r, \ell) \in \mathcal{G}(\mathcal{L})}U^{L}(G, r, \ell)$.

\begin{remark}
Recall from \cite[Lemma 2.6]{10.1145/506147.506153} that any instance $(G, r, \ell)$ with continuous and nondecreasing latency functions admits a Nash flow. Also, if $f^{NE}$ and $\tilde{f}^{NE}$ are Nash flows, then $C_i(f^{NE}) = C_i(\tilde{f}^{NE})$ for any commodity $i$. This fact, together with the property that all paths with positive flow in a Nash flow have equal latency, ensures that $\ell_Q(f^{NE})$ in (\ref{eq: UE def}) is well-defined.
\end{remark}

%%%%%%%%%%%%%% Average unfairness %%%%%%%%%%%%%%%%%%%
Consider an instance $(G, r, \ell)$ and a feasible flow $f$. For simplicity, consider a single commodity. If we consider a commuter traveling on a path $P \in \mathcal{P}$, their ex post unfairness is $\frac{\ell_P(f)}{\ell_Q(f)}$, where $Q$ is the minimum-latency path with a positive flow. Then, if we treat the flow $f$ as an ex ante route assignment, the expected unfairness of a commuter in the network will be $\frac{C(f)}{r\ell_Q(f)}$. This rationale gives rise to the following definition of the average unfairness measure.

\begin{definition}[Average unfairness]
Given an instance $(G, r, \ell)$ and a feasible flow $f$, the \emph{average unfairness} $U^{A}(G, r, \ell, f)$ is defined as
\begin{equation}
\label{eq: average}
U^{A}(G, r, \ell, f) = \max\left\{\frac{C_i(f)}{r_i\ell_Q(f)} | Q\in \mathcal{P}_i, f_Q > 0, i \in \{1, 2, ..., k\}\right\}.
\end{equation}
\end{definition}
We denote $U^{A}_i(G, r, \ell, f) = \max\{\frac{C_i(f)}{r_i\ell_Q(f)} | Q\in \mathcal{P}_i, f_Q > 0\}$ for commodity $i$'s average unfairness. 
For a given instance, we denote its average unfairness $U^{A}(G, r, \ell) := \sup_{f^{\star} \in OPT(G, r, \ell)}U^{A}(G, r, \ell, f^{\star})$.
Given a class of latency functions $\mathcal{L}$, we denote its average fairness
$U^{A}(\mathcal{L}) = \sup_{(G, r, \ell) \in \mathcal{G}(\mathcal{L})}U^{A}(G, r, \ell)$.

\begin{remark}
We view different commodities as incomparable and take the maximum across different commodities in (\ref{eq: average}). One can also define the average unfairness to be the average of different commodities' average unfairness. All the results in this work hold for both definition. %\young{Check this}
\end{remark}

\section{Unfairness Bounds and Comparisons}
\subsection{Unfairness Bounds}
In this section, we will prove the main result that gives the exact bound of average unfairness of a class of instances under general conditions. The bound is reminiscent of the results in \cite[Theorem 4.2]{doi:10.1287/opre.1070.0383} and \cite[Theorem 3.1]{10.5555/545381.545406}, which we unify here for completeness. As with those results, the bound depends on the \textit{steepness} of the class of latency functions.
\begin{definition}[Steepness]
    Let $\mathcal{L}$ denote a class of differentiable latency functions. For $\ell \in \mathcal{L}$, define the \emph{steepness} $\gamma(\ell)$ as $\gamma(\ell) = \sup_{x > 0} \frac{\hat{\ell}(x)}{\ell(x)}$. Define $\gamma(\mathcal{L}) = \sup_{\ell \in \mathcal{L}} \gamma(\ell)$. 
\end{definition}
We also consider the following set of assumptions.
\begin{assumption}
\label{assump: 1}
We assume that $\mathcal{L}$ is a family of differentiable and standard latency functions.
\end{assumption}

\begin{assumption}
\label{assump: 2}
We assume that $\mathcal{L}$ includes all the constant functions.
\end{assumption}

\begin{comment}
\begin{theorem}
\label{thm: 1}
Let $\mathcal{L}$ be a class of latency functions satisfying \cref{assump: 1} and \cref{assump: 2}. Then $U^{L}(\mathcal{L}) = U^{UE}(\mathcal{L}) = \gamma(\mathcal{L})$\footnote{Note two things. First, although the assumption of standard latency functions is not stated explicitly in \cite[Theorem 4.2]{doi:10.1287/opre.1070.0383}, it cannot be dispensed with. Second, only the single-commodity case is considered in \cite[Theorem 3.1]{10.5555/545381.545406}, but the result holds with multiple commodities.}.
\end{theorem}

\begin{remark}
\cref{assump: 1} alone suffices for $U^{L}(\mathcal{L}) \leq \gamma(\mathcal{L})$ and $U^{UE}(\mathcal{L}) \leq \gamma(\mathcal{L})$ to hold, but both assumptions are required for the reverse inequalities.
\end{remark}

\begin{theorem}
\label{thm: 2}
Let $\mathcal{L}$ be a family of latency functions satisfying \cref{assump: 1} and \cref{assump: 2}. Then $U^{A}(\mathcal{L}) = \gamma(\mathcal{L})$.
\end{theorem}
\end{comment}

\begin{theorem}
\label{thm: 2}
Let $\mathcal{L}$ be a class of latency functions satisfying \cref{assump: 1} and \cref{assump: 2}. Then $U^{L}(\mathcal{L}) = U^{UE}(\mathcal{L}) = U^A(\mathcal{L}) = \gamma(\mathcal{L})$.
\end{theorem}
\begin{proof}
The equalities $U^{L}(\mathcal{L}) = \gamma(\mathcal{L})$ and $U^{UE}(\mathcal{L}) = \gamma(\mathcal{L})$ are proved in  \cite[Theorem 4.2]{doi:10.1287/opre.1070.0383} and \cite[Theorem 3.1]{10.5555/545381.545406}, respectively. The equality for the average unfairness follows from \cref{lemma: 1} and \cref{lemma: lb} proved below.
\end{proof}

\begin{remark}
\cref{assump: 1} alone suffices for the three unfairness measures to be upper bounded by $\gamma(\mathcal{L})$, but both assumptions are required for the lower bound. 
\end{remark}

\begin{remark}
We highlight two points about the results in \cite[Theorem 4.2]{doi:10.1287/opre.1070.0383} and \cite[Theorem 3.1]{10.5555/545381.545406}. First, although the assumption of standard latency functions is not stated explicitly in \cite[Theorem 4.2]{doi:10.1287/opre.1070.0383}, it cannot be dispensed with. Second, while \cite[Theorem 3.1]{10.5555/545381.545406} only considers the single-commodity case, the result extends to multiple commodities.
\end{remark}

\begin{lemma}
\label{lemma: 1}
Let $\mathcal{L}$ be a family of latency functions satisfying \cref{assump: 1}. Then $U^{A}(\mathcal{L}) \leq \gamma(\mathcal{L})$.
\end{lemma}
\begin{proof}
Fix any instance $(G, r, \ell)$ and an optimal flow $f^{\star}$. Consider any $i \in \{1, 2, ..., k\}$. Since $f^{\star}$ is optimal, we get $\hat{\ell}_P(f^{\star}) \leq \hat{\ell}_Q(f^{\star})$ for any $P, Q \in \mathcal{P}_i$ with $f^{\star}_P > 0$. Let $Q$ be a minimum-latency path with a positive flow for commodity $i$, multiply both sides by $f^{\star}_P$, and sum over all paths:
\begin{equation}
\label{eq: opt}
\sum_{P\in \mathcal{P}_i} f^{\star}_P \hat{\ell}_P(f^{\star}) \leq (\sum_{P \in \mathcal{P}_i}f^{\star}_P)\cdot \hat{\ell}_Q(f^{\star}) = r_i \hat{\ell}_Q(f^{\star}).
\end{equation}

Combining the results, we get 

\begin{equation}
\label{eq: upper}
\begin{aligned}
C_i(f^{\star}) \leq \text{LHS of (\ref{eq: opt})} \leq \text{RHS of (\ref{eq: opt})} \leq \gamma(\mathcal{L})\cdot r_i \ell_Q(f^{\star}),
\end{aligned}
\end{equation}
where we use $\hat{\ell}_e(f^{\star}_e) = \ell_e(f^{\star}_e) + f^{\star}_e \ell_e'(f^{\star}_e) \geq \ell_e(f^{\star}_e)$ in the first inequality and $\hat{\ell}_Q(f^{\star}) = \sum_{e \in Q} \hat{\ell}_e(f_e^{\star}) \leq \gamma(\mathcal{L})\sum_{e \in Q} \ell_e(f_e^{\star}) = \gamma(\mathcal{L})\ell_Q(f^{\star})$ in the third inequality. Since (\ref{eq: upper}) holds for any commodity $i$, $U^{A}(\mathcal{L}) \leq \gamma(\mathcal{L})$.
\end{proof}

\begin{remark}
Another way to prove Lemma \ref{lemma: 1} is to invoke the inequality $U^L(G, r, \ell, f) \leq \gamma(\mathcal{L})$ in \cite[Theorem 4.2]{doi:10.1287/opre.1070.0383} %Theorem \ref{thm: 1} 
and the inequality $U^A(G, r, \ell, f) \leq U^L(G, r, \ell, f)$ in \cref{prop: UDJ and U}, which will be proved later without using Lemma \ref{lemma: 1}.
\end{remark}

In proving the lower bound, we make use of the Pigou network in \cref{fig:pigou}.
Before presenting the general lower bound in \cref{lemma: lb}, we provide a motivating example of the average unfairness of the optimal flow when latency functions are polynomials.

\begin{figure}
    \centering
    \begin{tikzpicture}[->, thick, node distance=2.5cm]
        %==============================
        % Left: Network
        %==============================
        % Nodes
        \node[circle, draw, minimum size=0.5cm, inner sep=1pt] (A) {$s$};
        \node[circle, draw, minimum size=0.5cm, inner sep=1pt, right=of A] (B) {$t$};
        
        % First Pigou network
        \draw (A) to[bend left=30] node[midway, above] {$\ell(x)$} (B);
        \draw (A) to[bend right=30] node[midway, below] {$c$} (B);
    \end{tikzpicture}
    \caption{Pigou network.}
    \label{fig:pigou}
\end{figure}
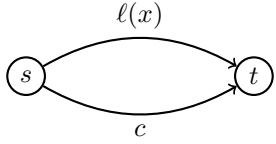

\begin{example}
\label{ex: lb}
Construct $(G, r, \ell)$ as follows: consider one unit of traffic ($r = 1$) traveling through the Pigou network in \cref{fig:pigou}, with $\ell(x) = x^{n}$ and $c = \epsilon$ where $\epsilon, n > 0$. Solving $(x^{n+1} + (1-x)\epsilon)' = 0$, we get $x = (\frac{\epsilon}{n+1})^{\frac{1}{n}}$, so 
\begin{equation}
U^A(G, r, \ell) \geq \frac{\frac{\epsilon}{n+1}(\frac{\epsilon}{n+1})^{\frac{1}{n}} + \epsilon(1 - (\frac{\epsilon}{n+1})^{\frac{1}{n}})}{\frac{\epsilon}{n+1}} \xrightarrow[\epsilon \rightarrow 0]{} n+1.
\end{equation}
\end{example}

\begin{lemma}\label{lemma: lb}
Let $\mathcal{L}$ be a family of latency functions satisfying \cref{assump: 1} and \cref{assump: 2}. Then $U^{A}(\mathcal{L}) \geq \gamma(\mathcal{L})$.
\end{lemma}
\begin{proof}
Consider $k = 1$. For any $\delta \in (0, \gamma(\mathcal{L}))$, there are $\ell \in \mathcal{L}$ and $\epsilon > 0$ such that $\hat{\ell}(\epsilon) \geq (\gamma(\mathcal{L})-\delta) \ell(\epsilon)$.

Consider $r$ units of traffic traveling through the Pigou network with latency functions $\hat{\ell}(\epsilon)$ and $\ell(x)$, respectively, and $r > \epsilon$. Solving $(x\ell(x) + \hat{\ell}(\epsilon)(r-x))' = 0$, we get $\hat{\ell}(x) = \hat{\ell}(\epsilon)$, so we assume $x = \epsilon$ below.

{\begin{equation}
U^A(\mathcal{L}) \geq \lim_{r \rightarrow \infty} \frac{\epsilon \ell(\epsilon) + \hat{\ell}(\epsilon)(r-\epsilon)}{r\ell(\epsilon)} = \lim_{r \rightarrow \infty} \frac{\epsilon}{r} + \frac{\hat{\ell}(\epsilon)(r-\epsilon)}{r\ell(\epsilon)}.
\end{equation}}

The first term goes to 0, so we focus on the second term.
\begin{equation}
\frac{\hat{\ell}(\epsilon)(r-\epsilon)}{r\ell(\epsilon)} \geq \frac{\hat{\ell}(\epsilon)(r-\epsilon)}{r\frac{\hat{\ell}(\epsilon)}{\gamma(\mathcal{L})-\delta}} = (1- \frac{\epsilon}{r})(\gamma(\mathcal{L}) - \delta) \xrightarrow[r \rightarrow \infty]{} \gamma(\mathcal{L}) - \delta.
\end{equation}

Letting $\delta \rightarrow 0$ finishes the proof.
\end{proof}

\begin{corollary}
If $\mathcal{L}$ is the set of polynomials with degree at most $n$ and nonnegative coefficients, $U^A(\mathcal{L}) = n+1$.
\end{corollary}
\begin{proof}
Since $(x \cdot x^n)' = (n+1)x^n$, the result follows from Theorem \ref{thm: 2}.
\end{proof}

\subsection{Comparisons}
\label{sec: comparisons}
It is informative to compare the values that the three unfairness measures take at the same flow. For the loaded and average unfairness measures, the following proposition characterizes this comparison.  
\begin{proposition}
\label{prop: UDJ and U}
Given any instance $(G, r, \ell)$ and feasible flow $f$, one of the following holds.
\begin{enumerate}
    \item $U^{L}(G, r, \ell, f) > U^A(G, r, \ell, f) > 1$.
    \item $U^{L}(G, r, \ell, f) = U^A(G, r, \ell, f) = 1$.
\end{enumerate}
\end{proposition}
% \begin{proof}
% See \cref{app: A}.
% \end{proof}
\begin{proof}
Fix an instance $(G, r, \ell)$ and drop the dependence. For any feasible flow $f$ and commodity $i$, since $r_i\max_{P\in \mathcal{P}_i, f_P > 0} \ell_P(f) \geq \sum_{P\in \mathcal{P}_i} f_P\ell_P(f)$, we have $U^{L}_i \geq U^A_i$, so $U^L \geq U^A$. Below, we show $U^{L} = U^A \iff U^A = 1$. 
\begin{equation}
\begin{aligned}
U^{L} = U^A &\iff \exists i^{\star} \in \{1, 2, ..., k\} \text{ s.t. } U^L_{i^{\star}} = U^A_{i^{\star}} \geq U^A\\
&\iff r_{i^{\star}}\max_{P\in \mathcal{P}_{i^{\star}}, f_P > 0} \ell_P(f) = \sum_{P\in \mathcal{P}_{i^{\star}}} f_P \ell_P(f)\\
&\iff \max_{P\in \mathcal{P}_{i^{\star}}, f_P > 0} \ell_P(f) = \min_{P\in \mathcal{P}_{i^{\star}}, f_P > 0} \ell_P(f) \\
&\iff U^A_{i^{\star}} = U^A = 1.
\end{aligned}
\end{equation}
\end{proof}
%
% When specializing $f$ to be an optimal flow, we obtain the following corollary. 
% %
% \begin{corollary}
% \label{cor: UDJ and U2}
% Given any instance $(G, r, \ell)$, one of the following holds.
% \begin{enumerate}
%     \item $U^{L}(G, r, \ell) > U^A(G, r, \ell) > 1$.
%     \item $U^{L}(G, r, \ell) = U^A(G, r, \ell) = 1$.
% \end{enumerate}
% \end{corollary}

A result similar to \cref{prop: UDJ and U} cannot be established between the UE unfairness and the loaded unfairness as there is no general relationship between the two measures (see \cref{ex:UL>UE=1} and \cref{ex:UL<UE}).
It is also possible that at a given flow, one of the measures equals to $1$ while the other is not (see \cref{ex:UL>UE=1} and \cref{ex:UL=1>UE}).
Moreover, as noted in \cite[Section 2.2]{doi:10.1287/opre.1040.0197}, the loaded unfairness is always greater than or equal to one, whereas UE unfairness can take any nonnegative value (see \cref{ex:UL=1>UE}).

\begin{example} 
Construct $(G, r, \ell)$ as follows: consider one unit of traffic ($r = 1$) traveling through the network in \cref{fig:pigou-series}. Consider the (optimal) flow $f^{\star}$ that routes half of the traffic on the upper path (path 1)  and the other half on the lower path (path 2). Then $U^{L}(G, r, \ell, f^{\star}) = 2 > U^{UE}(G, r, \ell, f^{\star}) = 1$.
\label{ex:UL>UE=1}
\end{example}

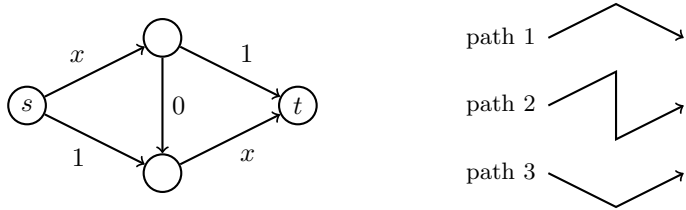
\begin{figure}[h!]
    \centering
    \begin{tikzpicture}[scale=0.9, ->, thick, node distance=2.5cm]
        %==============================
        % left: Network
        %==============================
        % Nodes
        \node[circle, minimum size=0.5cm, inner sep=1pt, draw] (S) at (0,1) {$s$};
        \node[circle, minimum size=0.5cm, inner sep=1pt, draw] (T) at (4,1) {$t$};
        \node[circle, minimum size=0.5cm, inner sep=1pt, draw] (V) at (2,2) {};
        \node[circle, minimum size=0.5cm, inner sep=1pt, draw] (W) at (2,0) {};
        
        % Edges
        \draw (S) -- node[midway, above left] {$x$} (V);
        \draw (S) -- node[midway, below left] {$1$} (W);
        \draw (V) -- node[midway, above right] {$1$} (T);
        \draw (W) -- node[midway, below right] {$x$} (T);
        \draw (V) -- node[midway, right] {$0$} (W);

        %==============================
        % Right: Paths
        %==============================
        \begin{scope}[xshift=7cm]
        % Path 1: x -> 1 (up)
        \node[black] at (0,2) {\small path 1};
        \coordinate (start1) at (0.7, 2);
        \coordinate (mid1) at (1.7, 2.5);
        \coordinate (end1) at (2.7,2);
        \draw[->, black] 
        (start1) to  (mid1) to  (end1);
        % Path 2: x -> e -> 1 (across)
        \node[black] at (0,1) {\small path 2};
        \coordinate (start2) at (0.7,1);
        \coordinate (mid2) at (1.7,1.5);
        \coordinate (mid22) at (1.7,0.5);
        \coordinate (end2) at (2.7,1);
        \draw[->, black] 
        (start2) to  (mid2) to  (mid22) to  (end2);
        % Path 3: 1 -> x (down)
        \node[black] at (0,0) {\small path 3};
        \coordinate (start3) at (0.7,0);
        \coordinate (mid3) at (1.7,-0.5);
        \coordinate (end3) at (2.7,0);
        \draw[->, black] 
        (start3) to  (mid3) to  (end3);
        \end{scope}
    \end{tikzpicture}
    \caption{Braess’s network.}
    \label{fig:Braess}
\end{figure}

\begin{example} 
Construct $(G, r, \ell)$ as follows: consider one unit of traffic ($r = 1$) traveling through the Braess’s network in \cref{fig:Braess}. Consider the (optimal) flow $f^{\star}$ that routes half of the traffic on the upper path (path 1)  and the other half on the lower path (path 3). Then $U^{L}(G, r, \ell, f^{\star}) = 1 > U^{UE}(G, r, \ell, f^{\star}) = \frac{3}{4}$.
\label{ex:UL=1>UE}
\end{example}

\begin{example}
Construct $(G, r, \ell)$ as follows: consider the multi-commodity network in \cref{fig:multi-network}. Consider the (optimal) flow $f^{\star}$ that routes all the traffic of the first commodity along the path $s_1-t$ and all the traffic of the second commodity along the only path $s_2-t$. Note that the Nash flow routes all the traffic of the first commodity through the path $s_1-s_2-t$. Thus, $U^{UE}(G, r, \ell, f^{\star}) = 3/2 > U^{L}(G, r, \ell, f^{\star}) = 1$.
\label{ex:UL<UE}
\end{example}
\begin{figure}[h!]
    \centering
    \begin{tikzpicture}[scale=0.80, ->, thick, node distance=2.5cm]
        % Nodes
        \node[circle, minimum size=0.5cm, inner sep=1pt, draw] (S1) at (0,2) {$s_1$};
        \node[circle, minimum size=0.5cm, inner sep=1pt, draw] (S2) at (2,0) {$s_2$};
        \node[circle, minimum size=0.5cm, inner sep=1pt, draw] (T)  at (4,2) {$t$};
        % Network edges
        \draw[->] (S1) -- node[midway, above] {$1.5$} (T) ;
        \draw[->] (S2) -- node[midway, below right] {$x$} (T);
        \draw[->] (S1) -- node[midway, below left] {$0$} (S2);
        % Demand arrows
        \draw[->] (-1.5,2) -- (S1) node[pos=0, left] {$r_1 = 0.25$};
        \draw[->] (0.5,0) -- (S2) node[pos=0, left] {$r_2 = 0.75$};        
    \end{tikzpicture}
    \caption{Multi-commodity network.}
    \label{fig:multi-network}
\end{figure}
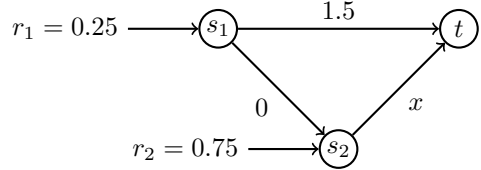

However, if we restrict our attention to optimal flows in single-commodity instances, we obtain a relationship similar to that in \cref{prop: UDJ and U}. 
Although both unfairness measures have been repeatedly studied in the literature for the past decade, to the best of our knowledge, this kind of relationship has never been theoretically established.
\begin{proposition}
\label{prop: UDJ and UTR}
Given any single-commodity ($k=1$) instance $(G, r, \ell)$ and optimal flow $f^{\star}$, one of the following holds.
\begin{enumerate}
    \item $U^{L}(G, r, \ell, f^{\star}) > U^{UE}(G, r, \ell, f^{\star})$.
    \item $U^{L}(G, r, \ell, f^{\star}) = U^{UE}(G, r, \ell, f^{\star}) = 1$.
\end{enumerate}
\end{proposition}
% \begin{proof}
% See \cref{app: B}.
% \end{proof}
\begin{proof}
Fix an instance $(G, r, \ell)$ and drop the dependence. For any optimal flow $f^{\star}$ and Nash flow $f^{NE}$, we have the relationship $\min_{P \in \mathcal{P}_i, f_P^{\star} > 0} \ell_P(f^{\star}) \leq \max_{P \in \mathcal{P}_i, f^{NE}_P > 0} \ell_P(f^{NE})$; for otherwise, $f^{NE}$ gives a lower total cost, a contradiction. Thus, $U^{L} \geq U^{UE}$.

Below, we show $U^{L} = U^{UE} \iff (U^{UE} = 1 \text{ and } U^L = 1)$. We only need to show the only if direction. 
\begin{equation}
\begin{aligned}
U^{L} = U^{UE} &\iff \min_{P \in \mathcal{P}_i, f_P^{\star} > 0} \ell_P(f^{\star}) = \max_{P \in \mathcal{P}_i, f^{NE}_P > 0} \ell_P(f^{NE})\\
&\Rightarrow  \max_{P \in \mathcal{P}_i, f_P^{\star} > 0} \ell_P(f^{\star}) = \min_{P \in \mathcal{P}_i, f_P^{\star} > 0} \ell_P(f^{\star}) \\
&\qquad\qquad\qquad\qquad\ \ \ = \max_{P \in \mathcal{P}_i, f^{NE}_P > 0} \ell_P(f^{NE})\\
&\Rightarrow  U^{UE} = 1 \text{ and } U^L = 1,
\end{aligned}
\end{equation}
where the second implication holds by contradiction; if it does not hold, $f^{\star}$ cannot be optimal.
\end{proof}

\begin{remark}
\cref{prop: UDJ and UTR} does not hold for multiple commodities (see \cref{ex:UL<UE}). 
However, it holds for another version of the definitions when we are allowed to compare the latencies of different commodities. Formally speaking, if we replace the condition $P, Q \in \mathcal{P}_i$ in (\ref{eq: L def}) and (\ref{eq: UE def}) with $P, Q \in \mathcal{P}$, \cref{prop: UDJ and UTR} holds for the modified definitions with multiple commodities.
\end{remark}

% \begin{example}
% Consider the network in Fig. \ref{fig: no multi}. The optimal flow is to route all the traffic of the first commodity through the $s_1$-$t$ path and all the traffic of the second commodity through the $s_2$-$t$ path. The Nash flow is the same for the first commodity, but it routes the second commodity through the $s_2$-$s_1$-$t$ path. Thus, $U^{L}(G, r, \ell, f^{\star}) = 1$ while $U^{UE}(G, r, \ell, f^{\star}) = 3/2$.
% \label{ex:UL-UE-multiple}
% \end{example}

% \begin{figure}
% \includegraphics[width=.5\textwidth]{figures/counterexample.jpg}
% \caption{\young{Please help me create a better figure. Also, we should use vector graphics (e.g. eps, pdf, but not jpg) whenever possible}} 
% \label{fig: no multi}
% \end{figure}

% \begin{corollary}
% \label{cor: UDJ and UTR}
% Given any single-commodity ($k=1$) instance $(G, r, \ell)$, one of the following holds.
% \begin{enumerate}
%     \item $U^{L}(G, r, \ell) > U^{UE}(G, r, \ell)$.
%     \item $U^{L}(G, r, \ell) = U^{UE}(G, r, \ell) = 1$.
% \end{enumerate}
% \end{corollary}

The comparisons also give rise to the following result, exemplified by Example \ref{ex: lb} and the example in \cite[Section 1]{10.5555/545381.545406}, where the upper bound for average unfairness and UE unfairness are approached, but not attained, when the respectively defined problem parameters $\epsilon$ goes to zero.

\begin{corollary}
\label{cor: no instance}
Let $\mathcal{L}$ be a family of latency functions.
\begin{enumerate}
    \item If $U^A(\mathcal{L}) > 1$, then $U^A(G,r,\ell, f^\star) < U^A(\mathcal{L})$ for any $(G,r,\ell) \in \mathcal{G}(\mathcal{L})$ and $f^\star \in OPT(G, r, \ell)$.
    \item If $k=1$ and $U^{UE}(\mathcal{L}) > 1$, then $U^{UE}(G,r,\ell, f^\star) < U^{UE}(\mathcal{L})$ for any $(G,r,\ell) \in \mathcal{G}(\mathcal{L})$ and $f^\star \in OPT(G, r, \ell)$.
\end{enumerate}
\end{corollary}
% \begin{proof}
% See \cref{app: C}.
% \end{proof}
\begin{proof}
We prove the first part. If an instance $(G, r, \ell)$ and optimal flow $f^\star$ attain the upper bound $U^A(\mathcal{L})$, then we have
\begin{equation}
U^{L}(G, r, \ell, f^\star) = U^A(G, r, \ell, f^\star) > 1,
\end{equation}
violating \cref{prop: UDJ and U}.
The second part follows similarly from \cref{prop: UDJ and UTR}.
\end{proof}

Finally, we show that there is no general relationship between the UE unfairness and average unfairness; specifically, they can be drastically different in different instances.
\begin{proposition}
\label{prop: UTR and U}
Given any instance $(G, r, \ell)$ and feasible flow $f$, there is no relationship between $U^{UE}(G, r, \ell, f)$ and $U^A(G, r, \ell, f)$.
\end{proposition}
% \begin{proof}
% See \cref{app: D}.
% \end{proof}
\begin{proof}
Consider Example \ref{ex: lb}, where $\lim_{\epsilon \rightarrow 0} U^{A}(G, r, \ell, f^\star) = n+1$ with the optimal flow $f^\star$. However, we have $U^{UE}(G, r, \ell, f^\star) = 1$ for every $\epsilon > 0$.
%Consider a modified version of Example \ref{ex: lb} where the Pigou network has latency functions $1$ and $mx^n$, respectively, where $m > 0$. Then $\lim_{m \rightarrow \infty} U^{A}(G, r, \ell) = n+1$. However, for every $m$, $U^{UE}(G, r, \ell) = 1$.
%
In contrast, consider the Pigou network in \cref{fig:pigou} with $\ell(x) = x^n$ and $c = (n+1)(1-\epsilon)$ (see \cite{10.5555/545381.545406}). Then $U^A(G, r, \ell, f^\star) = (n+1)-n(1-\epsilon)^{\frac{1}{n}} \xrightarrow[\epsilon \rightarrow 0]{} 1$, but $U^{UE}(G, r, \ell, f^\star) = \frac{(n+1)(1-\epsilon)}{1} \xrightarrow[\epsilon \rightarrow 0]{} n+1$.
\end{proof}

\section{Constrained System Optimum (CSO)}
% \subsection{General Comparisons}
In Section \ref{sec: comparisons}, we analyzed the worst-case unfairness for a class of instances across the different measures and compared the values of unfairness produced by these measures at a given instance and flow. 
A closely related question is that of the constrained system optimum problem first proposed in \cite{doi:10.1287/opre.1040.0197}. 
In this model, the goal is to find a feasible flow $f$ that minimizes the total latency with constraints on the unfairness it induces among the population. More formally, the CSO problem can be formulated as 
\begin{align}
    (CSO)\quad  \underset{f: U(f) \leq 1 + \tol}{\arg\min} C(f),
\end{align}
where $U$ is the unfairness measure of choice and $\tol \geq 0$ is the unfairness tolerance. 
%\bryce{hate to say it, but should we consider a variable other than epsilon for unfairness tolerance? Only saying this because we use the variable $\epsilon$ within examples as well as here. Could do $1+\epsilon = \beta$ like some others have done.}
%
Different unfairness measures give different solutions to this problem. 
In this section, we focus on comparing the solution sets of the CSO problem under loaded unfairness and average unfairness constraints. 
We denote by $OPT^L(G, r, \ell, \tol)$ and $OPT^A(G, r, \ell, \tol)$ the sets of optimal flows under the loaded unfairness and average unfairness constraints, respectively.

\begin{equation}
\begin{aligned}
OPT^L(G, r, \ell, \tol) &= \underset{f: U^L(f) \leq 1 + \tol}{\arg\min} C(f), \\\quad \quad \quad  
OPT^A(G, r, \ell, \tol) &= \underset{f: U^A(f) \leq 1 + \tol}{\arg\min} C(f).   
\end{aligned}
\end{equation}
Also, denote the cost of the respective sets by
\begin{equation}
\begin{aligned}
C^L(G, r, \ell, \tol) &= \underset{f: U^L(f) \leq 1 + \tol}{\min} C(f), \\\quad \quad \quad 
C^A(G, r, \ell, \tol) &= \underset{f: U^A(f) \leq 1 + \tol}{\min} C(f). 
\end{aligned}
\end{equation}
An immediate consequence of Proposition \ref{prop: UDJ and U} is the following corollary.
\begin{corollary}
\label{cor: better}
Given any instance $(G, r, \ell)$ and unfairness tolerance $\tol \geq 0$, $C^A(G, r, \ell, \tol) \leq C^L(G, r, \ell, \tol)$.
\end{corollary}

In some instances, the average unfairness constraint can strictly improve upon loaded unfairness; that is, \cref{cor: better} holds with strict inequality (see \cref{ex:ca<cl}). 
However, this is not always the case (see \cref{ex:ca=cl}). 

\begin{example} 
\label{ex:ca<cl}
Construct $(G, r, \ell)$ as follows: consider one unit of traffic ($r=1$) traveling through the Pigou network in \cref{fig:pigou} with $\ell(x) = x$ and $c = 1+\tol$, where $\tol \in (0, 1)$.
The optimal flow of this instance $f^{\star}$ routes $\frac{1+\tol}{2}$ on the upper edge and $\frac{1-\tol}{2}$ on the lower edge with $C(f^\star) = \frac{3+2\tol - \tol^2}{4} < 1$.
The optimal flow under the loaded unfairness $\{f^{L}\} = OPT^{L}(G, r, \ell, \tol)$ routes all traffic on the upper edge with unfairness of $U^L(f^{L}) = 1$ and cost $C^{L}(G, r, \ell, \tol) = 1$. However, for the average unfairness, we can improve the cost of $f^{L}$ by moving a small amount of traffic to the lower edge without violating the unfairness constraint. 

% Considering the average unfairness, we can improve the cost of $f^{L}$ by moving a small amount of traffic to the lower edge without violating the unfairness constraint. 
% %
% To see this, let $x$ denote the amount of flow on the upper edge. A feasible $x$ must satisfy
% \begin{equation}
% x^2 + (1+\epsilon)(1-x) \leq (1+\epsilon)x, 
% \end{equation}
% which holds for $x \in [1+\epsilon - \delta, 1+\epsilon + \delta]$, where $\delta = \sqrt{\epsilon^2 + \epsilon} > \epsilon$. Thus, we can indeed move a small amount of traffic to the lower edge and reduce the cost.
\end{example} 

\begin{example}
\label{ex:ca=cl}
Construct $(G, r, \ell)$ as follows: consider $3/4$ unit of traffic ($r=3/4$) traveling through the Braess's network in \cref{fig:Braess}. The optimal flow $f^{\star}$ distributes the traffic evenly across all three paths ($f^{\star}_{P_1} = f^{\star}_{P_2} = f^{\star}_{P_3} = \frac{1}{4}$) with a total cost of 1. Consider a small unfairness tolerance $\beta > 0$. Simple calculation shows that the flow $f^\beta$ defined by $f^\beta_{P_1} = f^\beta_{P_3} = 3/8$ and $f^\beta_{P_2} = 0$ is the unique optimal flow under both unfairness measures:
\begin{equation}
OPT^L(G, r, \ell, \beta) = OPT^A(G, r, \ell, \beta) = \{f^\beta\}.
\end{equation}
Therefore, while $f^\beta$ is not optimal in terms of total latency, and its loaded unfairness $U^L(G, r, \ell, f^\beta) = 1$ is strictly less than $1+\beta$, there is no feasible direction to improve upon it, even when considering average unfairness.
\end{example} 

\subsection{Sufficient Conditions for Strict Improvement}
We derive a sufficient condition for the constrained system optimal flow under average unfairness to achieve a strictly lower cost than that under loaded unfairness; that is, $C^A(G, r, \ell, \tol) < C^L(G, r, \ell, \tol)$.
This strict difference in cost is notable for two reasons: i) \cref{prop: UDJ and U} compares the loaded and average unfairness of the unconstrained optimal flow while here we compare the cost of two different flows with the same value of unfairness in each respective measure, thus providing greater nuance to the comparison of the two measures, and ii) should practitioners seek to solve the CSO problem in an effort to identify desirable flows, if the loaded unfairness measure is used, it can provide flows with greater envy across the population while also having higher cost.
The condition for the strict difference in cost can be checked by looking at the paths utilized under the solution $f\in OPT^{L}(G, r, \ell, \beta)$.

We first give some intuition. Consider a flow $f \in OPT^L(G, r, \ell, \tol)$ that is not a system optimal flow (i.e., $f\notin OPT(G,r, \ell)$). Then there exists an improvement direction, that is, a feasible flow $g$ such that $C(\delta g + (1-\delta) f) < C(f)\ \forall \delta \in (0, 1]$. This follows directly from the convexity of the (unconstrained) system optimum problem (under \cref{assump: 1}). One only needs to check if there exists $\delta \in (0,1]$ such that $h = \delta g + (1-\delta) f$ does not violate the average unfairness constraint $U^A(h) \leq 1+\beta$. This is equivalent to checking the continuity of $U^A$ around $f$, along with the fact that $U^{A}(f) < 1+\beta$ (see \cref{prop: UDJ and U}). 

We show that a sufficient condition for this to hold is if $f$ utilizes any of the minimum-latency paths for each commodity. To formalize this, we introduce the following notation for the set of utilized paths $P^{+}(f)$, the set of minimum-latency paths per commodity $P^{i}_{\min}(f)$, and the set of minimum-latency \emph{utilized} paths per commodity $P^{i}_{\min}(f)$:
\begin{align}
    &P^{+}(f)  := \{P\in \paths | f_P>0\}. \\
    &P^{i}_{\min}(f)  := \arg\min_{P\in \paths_i} \ell_{P}(f).\\
    &P^{i, +}_{\min}(f)  := \arg\min_{P\in \paths_i, f_P > 0} \ell_{P}(f).
\end{align}

\begin{proposition}
\label{cor: min latency}
Given any instance $(G, r, \ell)$ satisfying \cref{assump: 1} and unfairness tolerance $\tol > 0$, if there exists $f \in OPT^L(G, r, \ell, \tol)$ such that $P^{+}(f) \cap P^{i}_{\min}(f) \neq \emptyset$ for all $i \in \{1,2, \dots, k\}$, then one of the following holds
\begin{enumerate}
    \item $C^L(G, r, \ell, \tol) = C^A(G, r, \ell, \tol) = C^{OPT}(G, r, \ell)$.
    \item $C^L(G, r, \ell, \tol) > C^A(G, r, \ell, \tol) \geq C^{OPT}(G, r, \ell)$.
\end{enumerate}
\end{proposition}
\begin{proof} We state the proof for the single commodity case, which extends straightforwardly to instances with multiple commodities. 
If $C^L(G, r, \ell, \tol) > C^{OPT}(G, r, \ell)$, then there exists a feasible flow $g$ such that $C(\delta g + (1-\delta) f) < C(f), \forall  \delta \in (0,1]$. By \cref{prop: UDJ and U}, we know that $U^A(f) < 1 + \beta $. We show that $P^{+}(f) \cap P_{\min}(f) \neq \emptyset$ implies the continuity of the average unfairness measure $U^A$ around $f$, which then implies existence of $\delta$ small enough such that $U^A(\delta g + (1-\delta) f) \leq 1+ \beta$. 

By our assumption that $P^{+}(f) \cap P_{\min}(f) \neq \emptyset$, it follows that 
\begin{equation}
\min_{P \in P^+(f)} \ell_P(f) = \min_{P \in \mathcal{P}} \ell_P(f),
\end{equation}
and we denote this common value by $c^\star$. Let $h = \delta g + (1-\delta) f$. By continuity of the latency functions, the following hold when $\delta$ is sufficiently small.
\begin{enumerate}
    \item $\ell_P(h) < \ell_Q(h), \forall P \in P_{\min}(f), Q \notin P_{\min}(f)$.
    \item $P^+(h) \cap P_{\min}(f) \neq \emptyset$ (since $P^{+}(f) \cap P_{\min}(f) \neq \emptyset$).
\end{enumerate}

Thus, $P_{\min}^+(h) \subseteq P_{\min}(f)$, and we can make $\min_{P \in P^+(h)} \ell_P(h)$ arbitrarily close to $c^\star$.
\end{proof}
The condition in \cref{cor: min latency} need not hold in general networks: a CSO solution may not utilize a minimum-latency path like in \cref{ex:ca=cl}. It is, however, guaranteed to hold for parallel-link networks (i.e., single-commodity graphs with two nodes connected by a finite set of parallel edges) as we show in \cref{cor: strict improve}. 

\begin{corollary}
\label{cor: strict improve}
Given any parallel-link instance $(G, r, \ell)$ satisfying \cref{assump: 1} and unfairness tolerance $\tol > 0$, one of the following holds.
\begin{enumerate}
    \item $C^L(G, r, \ell, \tol) = C^A(G, r, \ell, \tol) = C^{OPT}(G, r, \ell)$.
    \item $C^L(G, r, \ell, \tol) > C^A(G, r, \ell, \tol) \geq C^{OPT}(G, r, \ell)$.
\end{enumerate}     
\end{corollary}
\begin{proof}
It suffices to show that any flow $f \in OPT^L(G, r, \ell, \tol)$ utilizes the minimum-latency path: $P^{+}(f) \cap P_{\min}(f) \neq \emptyset$. The rest follows from \cref{cor: min latency}.

Assume that there are $n$ edges with $\mathcal{P} = \{P_1, P_2, ..., P_n\}$, and the total flow is $r$. Suppose that a flow $f \in OPT^L(G, r, \ell, \tol)$ does not utilize any minimum-latency path; that is, $P^{+}(f) \cap P_{\min}(f) = \emptyset$, and let $P_1$ be one such unutilized minimum-latency path. 

We then construct a new flow $g$ by iteratively transferring flows from the maximum-latency paths to $P_1$ until one of the following conditions is met. For simplicity of notation, we abbreviate a path $P_j$ as $j$ in what follows.
\begin{enumerate}
    \item $\ell_1(g) = \arg\min_{j \in \{2, 3, ..., n\}, g_j > 0} \ell_j(g_j)$. 
    \item $g_1 = r$, meaning that all the flows are moved to $P_1$.
\end{enumerate}
Then, the transformed flow $g$ satisfies the loaded unfairness constraint since either all utilized paths have the same latency or the minimum latency is unchanged while the maximum latency decreases. Also, the following relationships hold.
\begin{enumerate}
    \item $\ell_1(g_1) \leq \arg\min_{j \in \{2, 3, ..., n\}, f_j > 0} \ell_j(f)$.
    \item $\ell_i(g_i) \leq \ell_i(f_i), \forall i \in \{2, 3, ..., n\}$.
\end{enumerate}

Let $y = f - g$ be the difference of the two flows. We will show that $C(f) > C(g)$. 
\begin{equation}
\begin{aligned}
&C(f) - C(g) \\
&= \sum_{i=2}^n f_i\ell_i(f_i) - \left(\sum_{i=2}^n g_i\ell_i(g_i) + (\sum_{i=2}^n (y_i))\cdot\ell_1(\sum_{j=2}^n (y_j)) \right) \\
&= \left(\sum_{i=2}^n g_i\cdot(\ell_i(f_i) - \ell_i(g_i))\right) + \left(\sum_{i=2}^n y_i\cdot(\ell_i(f_i) - \ell_1(\sum_{j=2}^n (y_j)))\right)\\
&> 0,
\end{aligned}
\end{equation}
where the first term is non-negative since the latency functions are increasing, and the second term is positive by our construction of the transformed flow $g$.
\end{proof}

The construction in \cref{cor: strict improve} is in fact more general. In particular, the same argument can be used to show the following stronger condition for parallel-link networks: 
\begin{equation}
\label{eq:parallel_cond}
\max_{P \in \mathcal{P}, f_P > 0} \ell_P(f_P) \leq \min_{P \in \mathcal{P}, f_P = 0} \ell_P(0)
\end{equation}
for any solution $f$ of the CSO problem under any of the three unfairness measures. This has an important implication on the set of used paths under any CSO flow as stated in \cref{cor:order_utilized}\footnote{The construction of $g$ does not require the CSO problem to be convex, so \cref{assump: 1} is not needed for \cref{cor:order_utilized}.}. 

%Consider the following ordering of the latency functions based on their zero-flow latencies ($\ell_1(0) \leq \ell_2(0) \leq ... \leq \ell_n(0)$). \Cref{eq:parallel_cond} implies that the set of utilized paths can only be of the form $P^{+}(f) = [m]$ for some $m\leq n$. 

\begin{corollary}
\label{cor:order_utilized}
Consider a parallel-link instance $(G, r, \ell)$ with $n$ edges. For any $\beta \geq 0$, choice of unfairness measure $U \in \{A, L, UE\}$ and solution $f \in OPT^U(G, r, \ell, \tol)$, there exists an ordering of the links such that $\ell_1(0) \leq \ell_2(0) \leq ... \leq \ell_n(0)$ for which the set of utilized paths has the form $P^{+}(f) = [m]$ for some $m\leq n$.  
\end{corollary}

\Cref{cor:order_utilized} extends the characterization in \cite{10.1145/380752.380783} for the set of utilized paths by the Nash and optimal flows to the CSO flows under all three unfairness measures and general latency functions (that are nonnegative, continuous, and nondecreasing). We note here that when multiple paths have identical zero-flow latencies, different optimal flows may require different orderings of the paths.

\subsection{Numerical Experiments} 
We conduct numerical experiments on general networks to compare the CSO cost under loaded and average unfairness. While strict improvement in the CSO cost under average unfairness is not guaranteed, our experiments show that it occurs across many benchmark networks.

To obtain approximate CSO solutions, we employ the algorithm proposed by \cite{jalota2023balancing}, which solves a sequence of convex programs whose objective interpolates between the potential function of the routing game and the total cost function: 
\begin{align}
    \alpha C(f) + (1-\alpha )\sum_{e\in E} \int_{0}^{f_e}\ell_e(s) d s.
\end{align}
The interpolation parameter $\alpha \in [0,1]$ controls the tradeoff between fairness and efficiency: $\alpha=0$ yields the Nash flow, while $\alpha=1$ yields the system optimum.

We vary $\alpha$ in increments of $0.01$ and compute a flow $f^{\alpha}$ for each value; this procedure is agnostic to the unfairness measure. For each resulting flow $f^{\alpha}$, we compute both loaded and average unfairness values. This allows us to trace the Pareto frontier between fairness and total cost for the two measures. For clarity of presentation, instead of absolute total cost values, we use the inefficiency ratio 
$\rho(f) = C(f^{\alpha}) / C(f^{\star})$
where $f^{\star} \in OPT(G, r, \ell)$. The implementation details follow \cite{jalota2023balancing}, except that our unfairness measures depend explicitly on path assignments (see \cref{remark:unfairness_unique}).  

\begin{table}[h!]
\centering
\begin{tabular}{lccc}
\toprule
\multicolumn{1}{c}{\textbf{Network Name}} & 
\multicolumn{3}{c}{\textbf{attributes}} \\
\cmidrule(lr){2-4}
 & $\lvert V \rvert$ & $\lvert E \rvert$ & $k$ \\
\midrule
Anaheim           & 416 & 914 & 1406 \\
Sioux Falls       & 24  & 76  & 528  \\
Massachusetts     & 74  & 258 & 1113 \\
Friedrichshain    & 224 & 523 & 506  \\
\bottomrule
\end{tabular}
\caption{Number of vertices $|V|$, edges $|E|$, and commodities $k$ in the used benchmark networks.}
\label{tab:network_info}
\end{table}

The experiments are conducted on four benchmark networks \cite{transportationnetworks2016} described in \cref{tab:network_info}. We use the Bureau of Public Roads (BPR) link latency functions \cite{sheffi1985urban}:
\[
l_e(f_e) = \xi_e \left( 1 + a \left( \frac{f_e}{\kappa_e} \right)^{4} \right), 
\]
where $\xi_e, \kappa_e$ are properties of the link and $a = 0.15$. The results are summarized in \cref{fig:city_plots}. Across all networks and unfairness levels, average unfairness consistently achieves a strictly lower cost than loaded unfairness, except for a single data point in the Friedrichshain network. These experiments illustrate the benefit of using the average unfairness measure as the user constraint in the CSO problem in real-world networks.

\begin{figure}[h!]
    \centering
    \includegraphics[width=0.6\linewidth]{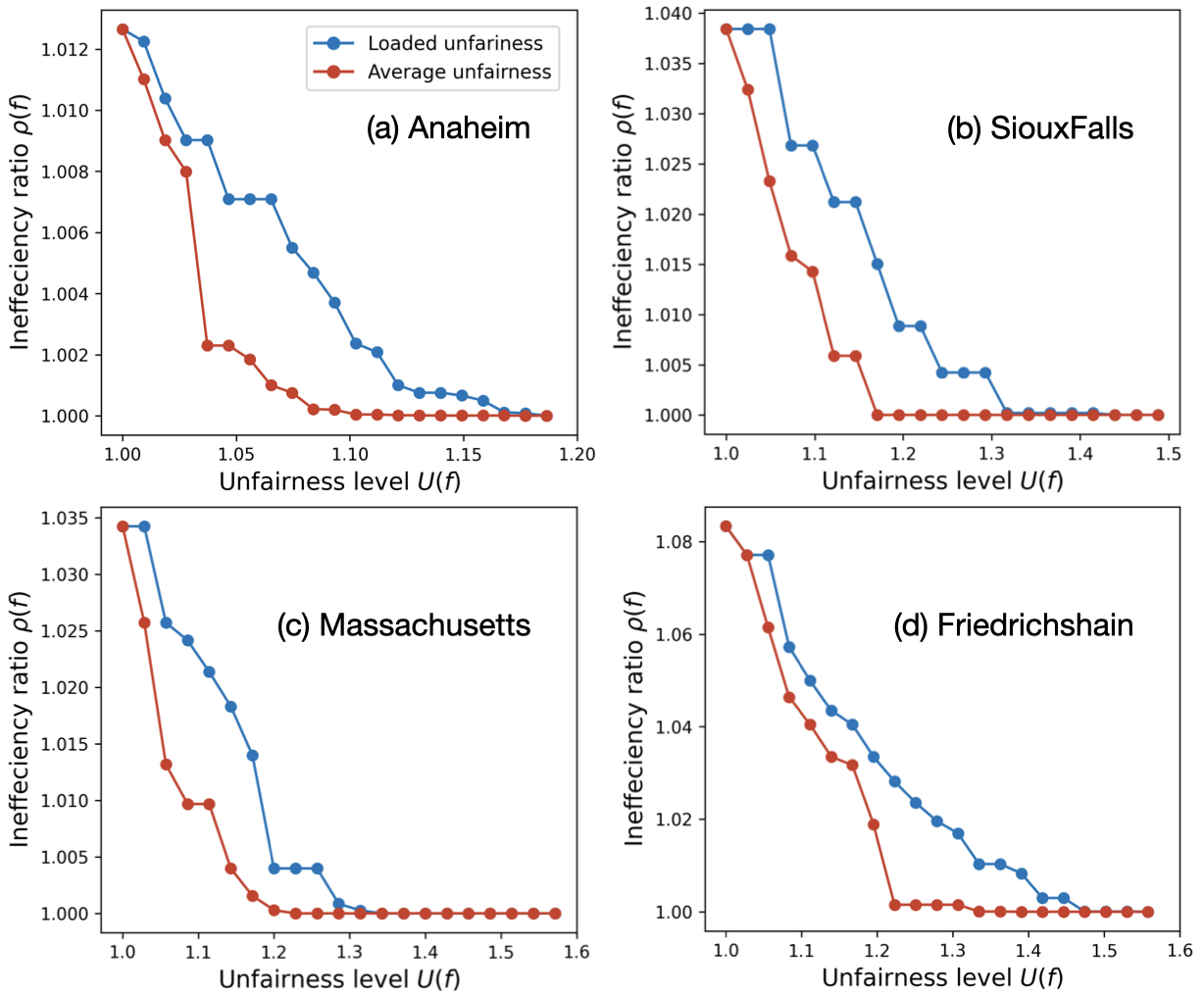}
    \caption{Comparison of CSO costs across varying unfairness levels under loaded (blue line) and average (red line) unfairness measures for four benchmark networks.}
    \label{fig:city_plots}
\end{figure}

\section{Conclusion and Future Work}
We introduced average unfairness as a new and intuitive measure of unfairness of feasible flows in routing games. 
We derived a tight worst-case bound on the unfairness of optimal flows under this new measure, and compared it to two commonly studied unfairness notions.
Our analysis also revealed implications for the constrained system optimum (CSO) problem, showing that average unfairness can lead to more efficient solutions than other fairness constraints.
%
%We provided the bound, compared it with the other two common measures, and derived some implications for the constrained system optimum (CSO) problem. 
%
An interesting direction for future work is to develop sharper conditions under which average unfairness leeds to strictly better cost than loaded unfairness in the CSO problem. 
We are also interested in studying bounds on the value of the CSO problem under average unfairness constraints, in the spirit of price of stability (POS) analyses (cf. \cite{christodoulou2011performance}).
Finally, we airm to explore broader applications of average unfairness in other contexts, like atomic congestion games, where fairness-efficiency trade-offs arise.

%%% The following command should be issued somewhere in the first column 
%%% of the final page of your paper.
\balance

%%%%%%%%%%%%%%%%%%%%%%%%%%%%%%%%%%%%%%%%%%%%%%%%%%%%%%%%%%%%%%%%%%%%%%%%

%%% The acknowledgments section is defined using the "acks" environment
%%% (rather than an unnumbered section). The use of this environment 
%%% ensures the proper identification of the section in the article 
%%% metadata as well as the consistent spelling of the heading.

%%%%%%%%%%%%%%%%%%%%%%%%%%%%%%%%%%%%%%%%%%%%%%%%%%%%%%%%%%%%%%%%%%%%%%%%

%%% The next two lines define, first, the bibliography style to be 
%%% applied, and, second, the bibliography file to be used.

\bibliographystyle{abbrv} 
\bibliography{refs}

\section*{Acknowledgments}
This research was supported in part by the Simons Foundation Agency project Collaborative Research: Transferable, Hierarchical, Expressive, Optimal, Robust, Interpretable NETworks (THEORINET) Award MPS-MODL-00814647. We used ChatGPT and Grammarly for language editing and proofreading of some paragraphs.

%%%%%%%%%%%%%%%%%%%%%%%%%%%%%%%%%%%%%%%%%%%%%%%%%%%%%%%%%%%%%%%%%%%%%%%%

\newpage
\appendix

\end{document}